\documentclass{llncs}
\usepackage{amsmath,amssymb}
\usepackage{graphicx,hyperref}
\usepackage{cite}
\usepackage{microtype}

\newif\iffull\fulltrue

\let\doendproof\endproof
\renewcommand\endproof{~\hfill\qed\doendproof}

\title{Triangle-Free Penny Graphs:\\ Degeneracy, Choosability, and Edge Count}
\author{David Eppstein\thanks{Supported in part by the National Science Foundation under Grants CCF-1228639, CCF-1618301, and CCF-1616248.}}
\institute{Department of Computer Science, University of California, Irvine}

\begin{document}
\maketitle

\begin{abstract}
We show that triangle-free penny graphs have degeneracy at most two, list coloring number (choosability) at most three, diameter $D=\Omega(\sqrt n)$, and at most $\min\bigl(2n-\Omega(\sqrt n),2n-D-2\bigr)$ edges.
\end{abstract}

\section{Introduction}

Penny graphs are the contact graphs of unit circles~\cite{HliKra-DM-01,PisRan-GaW-00} --- they are formed from non-overlapping sets of unit circles by creating a vertex for each circle and an edge for each tangency between two circles --- and as such, fit into a long line of graph drawing research on contact graphs of geometric objects~\cite{FraOssRos-CPC-94,BucGanPro-TALG-08,KlaNolUec-GD-15,Hli-DM-01,AlaEppKau-WADS-15}. The same graphs (except the graph with no edges) are also proximity graphs, the graphs determined from a finite set of points in the plane by adding edges between all closest pairs of points, and for this reason they are also called minimum-distance graphs~\cite{Csi-DCG-98,Swa-Geombinatorics-09}. A minimum-distance representation can be obtained from a contact representation by choosing a point at the center of each circle, and a contact representation can be obtained from a minimum-distance representation by scaling the points so their minimum distance is two and using each point as the center of a unit circle. However, finding either type of representation given only the graph is NP-hard, even for trees~\cite{BowDurLof-GD-15}.

As graph drawings, minimum distance representations are in many ways ideal: they have no crossings, all edges have unit length, and the angular resolution is at least $\pi/3$. Every graph that can be drawn with this combination of properties is a penny graph.
Moreover, penny graphs have \emph{degeneracy} at most three, where the degeneracy of a graph $G$ is the minimum number $d$ such that every subgraph of $G$ contains a vertex of at most $d$.
Equivalently, the vertices of any penny graph can be ordered so each vertex has at most three neighbors later than it in the ordering. This ordering leads to a linear-time greedy 4-coloring algorithm~\cite{HarRin-Pearl-8.4.8}, much simpler than known quadratic-time 4-coloring algorithms for arbitrary planar graphs~\cite{RobSanSey-STOC-96}. Additionally, although planar graphs with $n$ vertices can have $3n-6$ edges, penny graphs have at most  $\bigl\lfloor 3n-\sqrt{12n-3}\bigr\rfloor$ edges~\cite{Har-EdM-74}.
This bound is tight for pennies tightly packed into a hexagon~\cite{Kup-IG-94}, and its lower-order square-root term stands in an intriguing contrast to many similar bounds on the edge numbers of planar graphs, $k$-planar graphs, quasi-planar graphs, and minor-closed graph families, with constant or unknown lower-order terms~\cite{AckTar-JCTA-07,AgaAroPac-Comb-97,BraEppGle-GD-12,Epp-EJC-10,PacTot-Comb-97,SukWal-CGTA-15}.

Swanepoel~\cite{Swa-Geombinatorics-09} first considered corresponding problems for the \emph{triangle-free} penny graphs. In graph drawing terms, these are the graphs that can be drawn with no crossings, unit-length edges, and angular resolution strictly larger than $\pi/3$. Swanepoel observed that, as with triangle-free planar graphs more generally, an $n$-vertex triangle-free penny graph can have at most $2n-4$ edges.
As a lower bound, the square grids have
$\bigl\lfloor 2n-2\sqrt{n}\bigr\rfloor$
edges, as do some subsets of grids and some pentagonally-symmetric graphs found by Oloff de Wet~\cite{Swa-Geombinatorics-09}. Swanepoel conjectured that, of the two bounds, it is the lower bound that is tight.

Triangle-free planar graphs more generally have also been considered. Gr\"otzsch proved that these graphs are 3-colorable~\cite{Gro-WZMLU-59,Tho-JCTB-03} and they can be 3-colored in linear time~\cite{DvoKawTho-SODA-09}. However, not every triangle-free planar graph is  3-list-colorable: if each vertex is given a list of three colors, it is not always possible to assign each vertex a color from its list that differs from all its neighbors' assigned colors~\cite{Voi-DM-95}. 
3-list-colorability is known for bipartite planar graphs~\cite{AloTar-Comb-92}, planar graphs with girth at least five~\cite{Tho-JCTB-03}, and
planar graphs of girth four with well-separated 4-cycles~\cite{DvoLidSkr-SIDMA-10}, but these subclasses do not include all triangle-free penny graphs.

We continue these lines of research with the following new results.
\begin{itemize}
\item Every triangle-free penny graph with at least one cycle has at least four vertices of degree two or less. Consequently, the triangle-free penny graphs have degeneracy at most two.
\item Every triangle-free penny graph has list chromatic number (choosability) at most three,
and any list-coloring problem on a triangle-free penny graph with three colors per vertex can be solved in linear time.
\item Every $n$-vertex triangle-free penny graph has at most $2n-\Omega(\sqrt n)$ edges.
Thus, the form of Swanepoel's conjectured edge bound is correct, although we cannot confirm the conjectured constant factor on the square-root term.
\item Every penny graph has graph-theoretic diameter $\Omega(\sqrt{n})$, and every triangle-free penny graph with $n$ vertices and diameter $D$ has at most $2n-D-2$ edges. The combination of these two results provides an alternative proof of the $2n-\Omega(\sqrt{n})$ edge bound, but with a worse constant factor in the $\Omega$.
\end{itemize}

\section{Degeneracy}

We begin by showing that every triangle-free penny graph with at least one cycle has at least four vertices of degree two or less. It is convenient to begin with a special case of these graphs, the ones that are biconnected.

\begin{figure}[t]
\centering\includegraphics[width=3in]{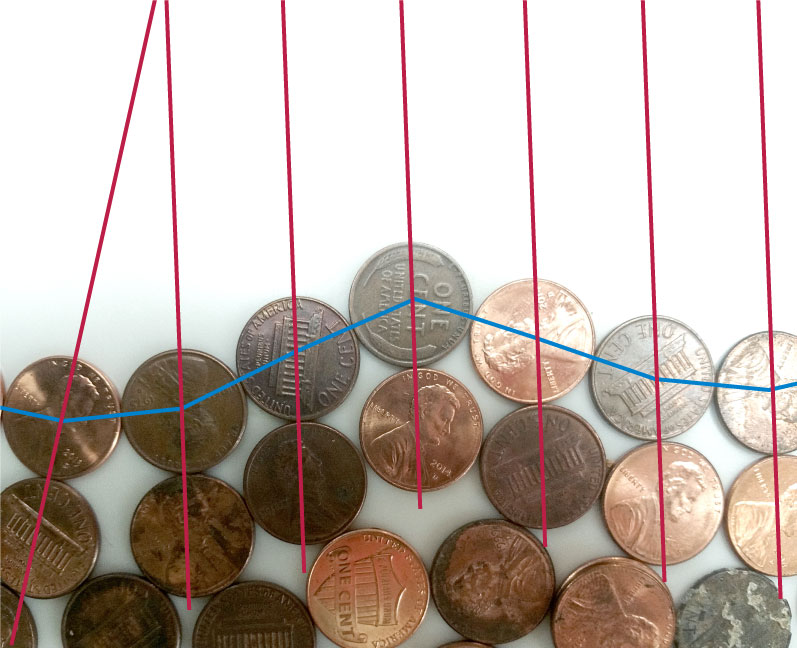}
\caption{Rays $R_v$ extending from the center of each boundary vertex directly away from the clockwise neighbor of its clockwise boundary neighbor, used in the proof of \autoref{lem:biconn-d2}.}
\label{fig:rays}
\end{figure}

\begin{lemma}
\label{lem:biconn-d2}
Every biconnected triangle-free penny graph has at least four vertices of degree two.
\end{lemma}

\begin{proof}
Given a biconnected triangle-free penny graph $G$, and its representation as a penny graph,
the outer face of the representation (as in any biconnected plane graph) consists of a simple cycle of vertices; in particular each vertex of this face has at least two neighbors. For each vertex $v$ of this simple cycle, let $w$ be the clockwise neighbor of $v$ in the cycle, and let $u$ be the neighbor of $v$ that is next in clockwise order around $v$ from $w$; define a ray $R_v$, having the center of the disk of $v$ as its apex, and pointing directly away from  the center of $u$ (\autoref{fig:rays}).
Given the same boundary vertices $v$ and $w$ in clockwise order, define the angle $\theta_w$ to be the angle made by rays $R_v$ and $R_w$, assigned a sign so that $\theta_w$ is positive if $R_w$ turns a clockwise angle (less than $\pi/2$) from $R_v$, and negative if $R_w$ turns counterclockwise with respect to $R_v$. If $R_v$ and $R_w$ are parallel, then we define $\theta_w=0$.
Then these rays and their angles have the following properties:
\begin{itemize}
\item Each ray $R_v$ points into the outer face of the drawing. Therefore, the sum of the turning angles of the rays as we traverse the entire outer face in clockwise order, $\sum\theta_v$, must equal $2\pi$.
\item If a boundary vertex $w$ has degree three or more, then $\theta_w\le 0$. For, if $v$ and $w$ are consecutive on the outer face, with $R_v$ pointing away from a neighbor $u$ of $v$ (as above) and $R_w$ pointing away from a neighbor $x$ of $w$, then the assumption that $w$ has degree at least three implies that $x\ne v$, and the assumption that $G$ is triangle-free implies that $x\ne u$.
If $x$ and $u$ touch, so that $uvwx$ forms a quadrilateral in $G$, then $R_v$ and $R_w$ are necessarily parallel, so $\theta_w=0$. In any other case, to prevent $x$ and $u$ from touching, $x$ must be rotated counterclockwise around $w$ from the position where it would touch $u$, causing angle $\theta_w$ to become negative.
\item At a boundary vertex $w$ of degree two, $\theta_w< 2\pi/3$. For, in this case, $R_w$ points away from $v$, the counterclockwise neighbor of $w$ on the outer face. Let $u$ be the neighbor of $v$ such that $R_v$ points away from $u$; then $w\ne u$. Because both $R_v$ and $R_w$ belong to lines through the center of $v$, their angle $\theta_w$ is complementary to angle $wvu$, which must be greater than $\pi/3$ in order to prevent circles $u$ and $w$ from overlapping or touching (and forming a triangle). Therefore, $\theta_w$ is less than $2\pi/3$.
\end{itemize}
For the sequence of angles $\theta_w$, each less than $2\pi/3$, to add to a total angle of $2\pi$, there must be at least four positive angles in the sequence, and therefore there must be at least four degree-two vertices.
\end{proof}

\begin{theorem}
\label{thm:many-d2}
Every triangle-free penny graph $G$ with at least one cycle has at least four non-articulation vertices of degree two or less.
\end{theorem}

\begin{proof}
By the assumption that $G$ has at least one cycle, it has at least one nontrivial biconnected component $C$. By \autoref{lem:biconn-d2}, $C$ has at least four degree-two vertices, each of which either has degree two in $G$ or forms an articulation point of $G$. If it forms an articulation point, then the tree of biconnected components connected through it to $G$ has at least one leaf, which must either be a vertex of degree one in $G$ or a nontrivial biconnected component with at least four degree-two vertices, only one of which can be an articulation point. Thus, each of the four degree-two vertices in $C$ is either itself a non-articulation vertex of degree at most two in $G$ or leads to such a vertex.
\end{proof}

The bound on the number of degree-two vertices is tight for square grids.

\begin{theorem}
The degeneracy of every triangle-free penny graph is at most two.
\end{theorem}

\begin{proof}
Every subgraph of a triangle-free penny graph is another triangle-free penny graph, so the result follows from \autoref{thm:many-d2} and from the fact that, in a graph with no cycles (a forest) there always exists a vertex of degree one or less (a leaf or an isolated vertex).
\end{proof}

\section{Choosability}

The \emph{choosability}, or \emph{list chromatic number}, of a graph $G$ is the minimum number $c$ such that, for every labeling of each vertex of $G$ by a list of $c$ colors,  there  exists an assignment of a single color from its list to each vertex, with no two adjacent vertices assigned the same color. The usual graph coloring problem is a special case in which all vertices have the same list. Known relations between list coloring and graph degeneracy~\cite{AloTar-Comb-92} give us the following result:

\begin{theorem}
If a triangle-free penny graph is labeled by a list of three colors for each vertex, then we can find a solution to the list coloring problem for the resulting labeled graph in linear time. The algorithm needs as input only the abstract graph, not its representation as a penny graph.
\end{theorem}

\begin{proof}
Find a vertex of degree at most two, remove it from the graph, color the remaining subgraph recursively, and put back the removed vertex. It has at most two neighbors, preventing it from being assigned at most two colors from its list of three colors, so there always remains at least one color available for it to use.

Linear time follows by maintaining the degree of each vertex in the reduced graph formed by the  removals, a list of vertices of reduced degree at most two, and a stack of removals to be reversed. It takes constant time per vertex removal and replacement to update these data structures.
\end{proof}

\begin{corollary}
Triangle-free penny graphs have choosability at most three.
\end{corollary}

This bound is tight as the odd cycles of length $\ge 5$ are triangle-free penny graphs with choosability exactly three.

\section{Edge count}

We derive a  bound on the number of edges of a triangle-free penny graph by using the isoperimetric theorem to show that the outer face of any representation as a penny graph has many vertices, and then by using Euler's formula to show that a planar graph with a large face has few edges.

\begin{lemma}
\label{lem:voronoi-area}
Let $v$ be a vertex of a penny graph that (in some representation of the graph as a penny graph) is not on the outer face. Then, in the Voronoi diagram of the centers of the circles in the representation, the Voronoi cell containing $v$ has area at least $2\sqrt{3}$, which is the area of a regular hexagon circumscribed around a unit circle.
\end{lemma}

\begin{proof}
The area is minimized when each Voronoi neighbor of $v$ is as close as possible to $v$ (so that the neighbor's circle touches that of $v$, causing the Voronoi cell of $v$ to circumscribe its circle), when the neighbors are equally spaced around $v$ (forming a regular polygon), and when the number of neighbors is as large as possible (forming a hexagon). The first two of these claims follow from the fact that any other configuration of neighbors can be continuously deformed to make the area of $v$'s cell smaller, while the last one follows by comparing the areas of the other possible regular polygons.
\end{proof}

\begin{lemma}
\label{lem:penny-has-big-face}
In any penny graph representation of a graph $G$ with $n$ vertices,
the number of vertex-face incidences on the outer face of the representation is at least
\[
\sqrt{\pi\cdot 2\sqrt{3}\cdot n} - O(1)\approx 3.3\sqrt{n}.
\]
\end{lemma}

\begin{proof}
Unless there are at least this many incidences, by \autoref{lem:voronoi-area} there must be a total area of at least $2\sqrt{3}\cdot n - O(\sqrt n)$ enclosed by the outer face, because each Voronoi cell of an inner vertex is enclosed and the Voronoi cells are all disjoint. The result follows from the facts that each vertex-face incidence accounts for $2$ units of length of the outer face (the two radii of a single unit circle in the representation, along which the outer face enters and then leaves that circle) and that any curve that encloses area $A$ must have length at least $2\sqrt{\pi A}$ (the isoperimetric theorem, with the shortest enclosing curve being a circle).
\end{proof}

\begin{lemma}
\label{lem:big-face-few-edges}
Let $G$ be an $n$-vertex triangle-free plane graph in which one face has $k$ vertex-face incidences. Then $G$ has at most $2n-k/2-2$ edges.
\end{lemma}

\begin{proof}
Vertex-face incidences and edge-face incidences on any face are equal, so the same face of $G$ that has $k$ vertex-face incidences also has $k$ edge-face incidences. We count the number of edge-face incidences in $G$ in two ways: by counting two incidences for each edge,
and by summing the lengths of the faces. Each face of $G$ has at least four edges, so if there are $e$ edges and $f$ faces then we have the inequality $2e\ge 4(f-1)+k$,
or equivalently
$e/2 - k/4 + 1\ge f$.
Using this inequality to replace $f$ in Euler's formula $n-e+f=2$, we obtain
$n-e+e/2-k/4+1\ge 2$,
or equivalently
$e\le 2n-k/2-2$
as claimed.
\end{proof}

\begin{theorem}
The number of edges in any $n$-vertex triangle-free penny graph is at most
\[
2n-\frac{1}{2}\sqrt{\pi\cdot 2\sqrt{3}\cdot n} + O(1)\approx 2n-1.65\sqrt{n}.
\]
\end{theorem}

\begin{proof}
\autoref{lem:penny-has-big-face} proves the existence of a large face, and plugging the size of this face as the variable~$k$ in \autoref{lem:big-face-few-edges} gives the stated bound.
\end{proof}

We leave the problem of closing the gap between this upper bound and Swanepoel's $2n-2\sqrt{n}$ lower bound as open for future research.

\section{Diameter}

Our results on degeneracy and number of edges can be connected via the following two results, which provide an alternative proof that the number of edges in a triangle-free penny graph is $2-\Omega(\sqrt{n})$.

\begin{theorem}
\label{thm:high-diam}
Every connected $n$-vertex penny graph has diameter $\Omega(\sqrt{n})$.
\end{theorem}

\begin{proof}
By a standard isodiametric inequality~\cite{Bie-JDMV-15}, for the convex hull of $n$ disjoint unit disks to enclose area $2\pi n$, it must have (geometric) diameter $\Omega(\sqrt{n})$. In order to connect two unit disks at geometric distance $\Omega(\sqrt n)$ from each other, they must also be at graph-theoretic distance $\Omega(\sqrt n)$.
\end{proof}

\begin{theorem}
\label{thm:diam-edges}
Every connected $n$-vertex triangle-free penny graph $G$ with diameter $D$ has at most $2n-D-2$ edges.
\end{theorem}

\begin{proof}
We use induction on $n$.
If $G$ has no cycle, it is a tree, with $n-1$ edges, and the result follows from the fact that $D\le n-1$. Otherwise, let $uw$ be a diameter pair, and let $v$ be any vertex of degree at most two, whose removal does not disconnect $G$, distinct from $u$ and $w$. The existence of $v$ follows from \autoref{thm:many-d2}. Then $G-v$ has one less vertex, one or two fewer edges, and diameter at least $D$. The result follows by applying the induction hypothesis to $G-v$.
\end{proof}

It is not true more generally that 2-degenerate triangle-free planar graphs with diameter $D$ have at most $2n-D-2$ edges; \autoref{thm:diam-edges} relies on the specific properties of triangle-free penny graphs. However, we can prove analogous bounds of $2n-\Omega(\sqrt n)$ and $2n-D-2$ on the numbers of edges in \emph{squaregraphs}~\cite{BanCheEpp-SJDM-10}, plane graphs in which every bounded face is a quadilateral and every vertex that does not belong to the unbounded face has degree at least four. The details are given in
\iffull
the appendix.
\else
the full version~\cite{XXX}
\fi

\iffull
\bibliographystyle{unabuser}
\else
\bibliographystyle{splncs}
\fi
\bibliography{penny}

\iffull
\newpage
\appendix
\section{Analogous results for squaregraphs}

\begin{figure}[t]
\centering\includegraphics[width=3in]{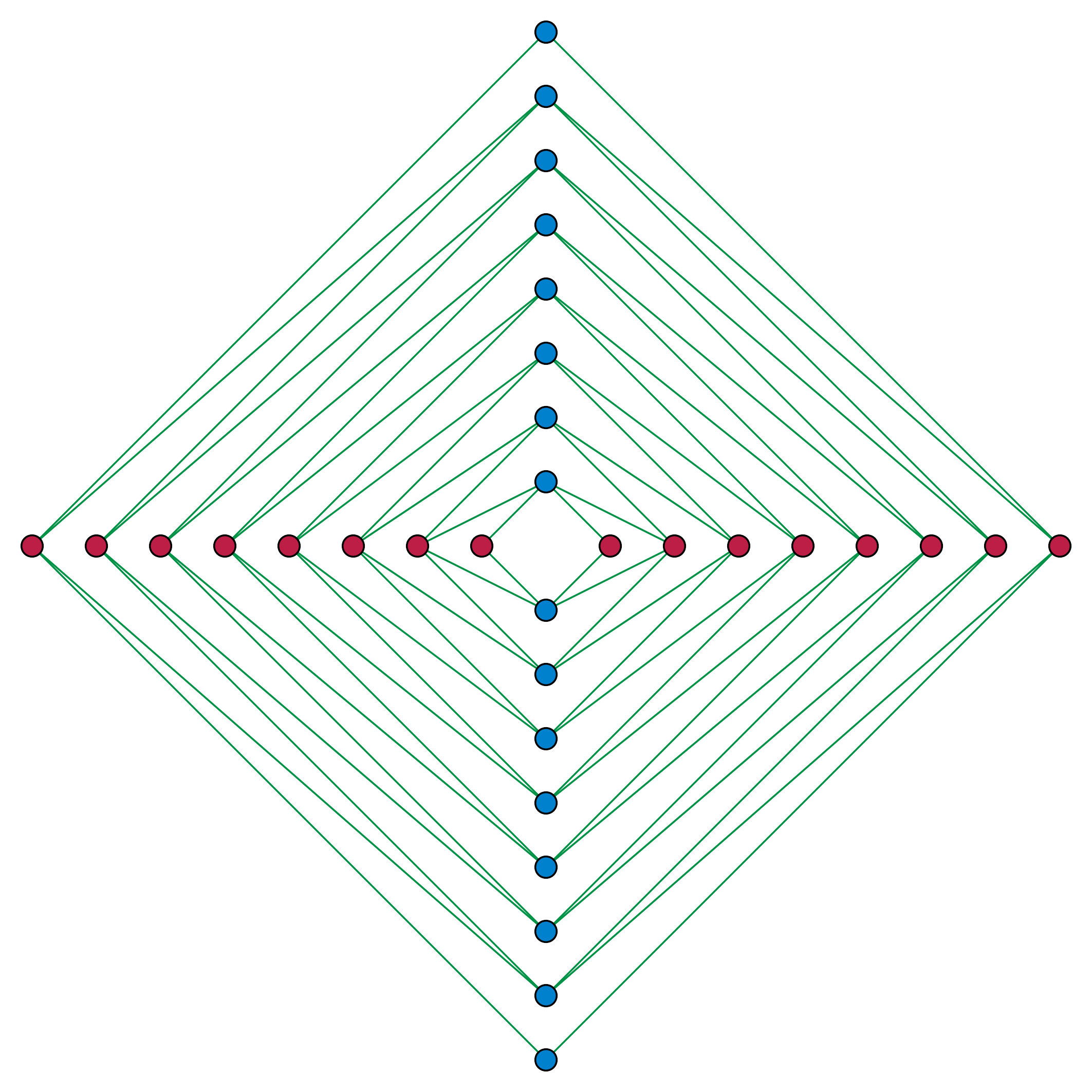}
\caption{One of a family of 2-degenerate bipartite planar graphs with $2n-4$ edges.}
\label{fig:malyshev32}
\end{figure}

We have shown that triangle-free penny graphs are 2-degenerate, have at most $2n-\Omega(\sqrt n)$ edges, and have at most $2n-D-2$ edges where $D$ is the graph-theoretic diameter of the graph.
These results on numbers of edges do not generalize to 2-degenerate triangle-free planar graphs more generally; \autoref{fig:malyshev32} shows the construction for a family of $2$-degenerate triangle-free planar graphs (actually bipartite planar permutation graphs) with unbounded diameter
and $2n-4$ edges, the maximum possible for any triangle-free planar graph.

However, as we now show, analogous bounds on edge number do apply to another class of triangle-free planar graphs, the \emph{squaregraphs}. These are the plane graphs in which every bounded face is a quadilateral and every vertex that does not belong to the unbounded face has degree at least four~\cite{BanCheEpp-SJDM-10}. Not every triangle-free penny graph is a squaregraph, and not every square-graph is a triangle-free penny graph; nevertheless, these two classes of graphs behave similarly in many respects.

Importantly, the squaregraphs obey a lemma corresponding to \autoref{lem:biconn-d2}: every biconnected squaregraph has at least four vertices of degree two~\cite[Proposition 4.1]{BanCheEpp-SJDM-10}. From this it immediately follows that they are 2-degenerate and 3-choosable. This also gives us the analogues of \autoref{thm:many-d2} (every squaregraph with at least one cycle has at least four non-articulation vertices of degree two) and \autoref{thm:diam-edges} (every connected squaregraph with diameter $D$ has at most $2n-D-2$ edges), as the proofs of these results for triangle-free penny graphs use only \autoref{lem:biconn-d2}. However, squaregraphs do not obey an analogue of \autoref{thm:high-diam}: arbitrarily large squaregraphs can have bounded diameter. Therefore, we cannot use the analogue of \autoref{thm:diam-edges} to obtain the claimed $2n-\Omega(\sqrt n)$ bound on the numbers of edges of squaregraphs.

Instead, we use the fact that squaregraphs are dual to hyperbolic line arrangements in which no three lines all cross each other~\cite[Theorem 6.1]{BanCheEpp-SJDM-10}.
In a hyperbolic arrangement with $\ell$ lines and $c$ crossings, the number of squaregraph vertices (dual to cells of the arrangement) is $c+\ell+1$ and the number of squaregraph edges (dual to the line segments between cells in the arrangement) is $2c+\ell$.
Therefore, to construct a squaregraph with the maximum number of edges for a given number $n=c+\ell+1$ of vertices,
we need to maximize $c$ and correspondingly minimize $\ell$.
However, because the intersection graph of the lines is triangle-free, it follows by Tur\'an's theorem that there will be at most $c\le\lfloor\ell/2\rfloor\cdot\lceil\ell/2\rceil$ crossings of lines.
Combining this inequality with the formulas for the numbers of edges and vertices in a squaregraph proves that the number of edges in any $n$-vertex squaregraph
is at most $\lfloor 2n-2\sqrt n\rfloor$.
This bound is tight, as it can be achieved for any $n$ by finding the smallest square grid with at least $n$ vertices and then removing degree-two vertices until the number of remaining vertices is~$n$.

\fi

\end{document}